\documentclass[11pt,reqno]{amsart}
\usepackage{latexsym,amssymb,amsfonts,mathrsfs}
\usepackage{amsthm}
\usepackage{amsmath}
\usepackage[a4paper,bottom=2.5cm]{geometry}
\usepackage{a4wide}
\usepackage{enumerate}
\usepackage{comment}
\usepackage[latin1, utf8]{inputenc}
\usepackage{eufrak}
\usepackage[usenames]{color}
\usepackage{graphicx}
\usepackage{lmodern,xfrac,xcolor}
\usepackage[noadjust]{cite}
\usepackage[unicode]{hyperref}

\usepackage{mathtools}
\usepackage[toc]{appendix}
\mathtoolsset{showonlyrefs}
\usepackage{scalerel}

\def\e{\mathrm{e}}
\def\ve{\varepsilon}
\def\la{\lambda}
\def\me{\mathsf{e}}
\def\mv{\mathsf{v}}
\def\mE{\mathsf{E}}
\def\mV{\mathsf{V}}

\def\ea{\EuFrak{a}}

\def\Ker{\mathrm{Ker}}
\def\ran{\mathrm{ran}}
\def\ve{\varepsilon}
\def\diag{\mathrm{diag}}

\def\real{\mathbb{R}}
\def\nat{\mathbb{N}}

\def\mcH{\mathcal{H}}

\def\mcY{\mathcal{Y}}

\def\dx{\, dx}
\def\dt{\, dt}
\def\ds{\, ds}

\DeclareMathOperator{\Tr}{Tr}
\DeclareMathOperator{\HS}{HS}

\newtheorem{theorem}{Theorem}

\newtheorem{proposition}[theorem]{Proposition}
\newtheorem{corollary}[theorem]{Corollary}
\newtheorem{definition}[theorem]{Definition}

\newtheorem{remark}[theorem]{Remark}

\newtheorem{example}[theorem]{Example}

\numberwithin{equation}{section} \numberwithin{theorem}{section}

\begin{document}

\title{On the parabolic Cauchy problem for quantum graphs with vertex noise}

\author{M.~Kov\'acs}
\address{Faculty of Information Technology and Bionics\\P\'azm\'any P\'eter Catholic University,
Budapest, Hungary and
Department of Differential Equations, Faculty of Natural Sciences, Budapest University of Technology and Economics,
Budapest, Hungary and
Chalmers University of Technology and University of Gothenburg,
Gothenburg, Sweden}
\email{mihaly@chalmers.se}
\thanks{}

\author{E.~Sikolya}
\address{Department of Applied Analysis and Computational Mathematics\\
E\"otv\"os Lor\'and University\\
Budapest, Hungary\\
Alfr\'ed R\'enyi Institute of Mathematics,\\ 
Budapest, Hungary}
\email{eszter.sikolya@ttk.elte.hu}

\date{\today}


\begin{abstract}
We investigate the parabolic Cauchy problem associated with quantum graphs including Lipschitz or polynomial type nonlinearities and additive Gaussian noise perturbed vertex conditions. The vertex conditions are the standard continuity and Kirchhoff assumptions in each vertex. In the case when only Kirchhoff conditions are perturbed, we can prove existence and uniqueness of a mild solution with continuous paths in the standard state space $\mcH$ of square integrable functions on the edges. We also show that the solution is Markov and Feller. Furthermore, assuming that the vertex values of the normalized eigenfunctions of the self-adjoint operator governing the problem are uniformly bounded, we show that the mild solution has continuous paths in the fractional domain space associated with the Hamiltonian operator, $\mcH_{\alpha}$ for $\alpha<\frac{1}{4}$. This is the case when the Hamiltonian operator is the standard Laplacian perturbed by a potential. We also show that if noise is present in both type of vertex conditions, then the problem admits a mild solution with continuous paths in the fractional domain space $\mcH_{\alpha}$ with $\alpha<-\frac{1}{4}$ only. These regularity results are the quantum graph analogues obtained by da Prato and Zabczyk \cite{DPZ93} in case of a single interval and classical boundary Dirichlet or Neumann noise.
\end{abstract}

\maketitle

\pagestyle{plain}

%%%%%%%%%%%%%%%%%%%%%%%%%%%%%%%%%%%%%%%%%%%%
\section{Introduction}
%%%%%%%%%%%%%%%%%%%%%%%%%%%%%%%%%%%%%%%%%%%%%

We consider a so-called quantum graph; that is, is a metric graph $G$, equipped with a diffusion operator on each edge and certain vertex conditions. Our terminology follows \cite[Chap.~1]{BeKu} (see also \cite{Mu14}), we list here only the most important concepts. The graph $G$ consists of a finite set of vertices $\mV = \{\mv\}$ and a finite set $\mE = \{\me\}$ of edges connecting the vertices. We denote by $m=|\mE|$ the number of edges and by  $n=|\mV|$ the number of vertices. In general, a metric graph is assumed to have directed edges; that is edges having an origin and a terminal vertex. In our case, dealing with self-adjoint operators, we can just consider undirected edges. Each edge is assigned a positive length $\ell_{\me}\in (0,+\infty)$, and we denote by $x\in [0,\ell_{\me}]$ a coordinate of $G$. We assume that $G$ is simple; that is, there are no multiple edges connecting two vertices, and there are no loops at any of the vertices in $G$. 

The metric graph structure enables one to speak about functions $u$ on $G$, defined along the edges such that for any coordinate $x$, the function takes its value $u(x)$. If we emphasize that $x$ is taken from the edge $\me$, we write $u_{\me}(x)$. Thus, a function $u$ on $G$ can be regarded as a vector of functions that are defined on the edges, therefore we will also write
\[u=\left(u_{\me}\right)_{\me\in \mE},\]
and consider it as an element of a product function space.

To write down the vertex conditions in the form of equations, for a given function $u$ on $G$ and for each $\mv\in\mV$, we introduce the following notation. For any $\mv\in\mV$, we denote by $\mE_{\mv}$ the set of edges incident to the vertex $\mv$, and by $d_{\mv}=|\mE_{\mv}|$ the degree of $\mv$. Let $u_{\me}(\mv)$ denote the value of $u$ in $\mv$ along the edge $\me$ in the case $\me\in \mE_{\mv}$. Let $\mE_{\mv}=\{\me_1,\dots ,\me_{d_{\mv}}\}$, and define
\begin{equation}\label{eq:Fv}
U(\mv)=\left(u_{\me}(\mv)\right)_{\me\in \mE_{\mv}}=\begin{pmatrix}
	u_{\me_1}(\mv)\\
	\vdots\\
	u_{\me_{d_{\mv}}}(\mv)
\end{pmatrix}\in \real^{d_{\mv}},
\end{equation}
the vector of the function values in the vertex $\mv$. 

Let $I_{\mv}$ be the bi-diagonal matrix
\begin{equation}\label{eq:Iv}
I_{\mv}=\begin{pmatrix}
	1 & -1 & & \\
	& \ddots & \ddots & \\
	& & 1 & -1 
	\end{pmatrix}\in\real^{(d_{\mv}-1)\times d_{\mv}}.
\end{equation}
It is easy to see that if we set
\begin{equation}\label{eq:contv}
I_{\mv}U(\mv)=0_{\real^{d_{\mv}-1}},
\end{equation} 
this means that all the function values coincide in $\mv$.  If it is satisfied for each vertex $\mv\in\mV$ for a function $u$ on $G$ which is continuous on each edge, including the one-sided continuity at the endpoints, then we call $u$ is continuous on $G$. 

Similarly, for a function $u$ on $G$ which is differentiable on each edge; that is. $u'_{\me}$ exists for each $\me\in\mE$ including the one-sided derivatives at the endpoints, we set
\begin{equation}\label{eq:Fvv}
U'(\mv)=\left(u'_{\me}(\mv)\right)_{\me\in\mE_{\mv}}=\begin{pmatrix}
	u'_{\me_1}(\mv)\\
	\vdots\\
	u'_{\me_{d_{\mv}}}(\mv)
\end{pmatrix}\in \real^{d_{\mv}},
\end{equation}
the vector of the function derivatives in the vertex $\mv$. We will assume throughout the paper that derivatives are taken in the directions away from the vertex $\mv$ (i.e.~into the edge), see \cite[Sec.~1.4.]{BeKu}.

First we aim to analyse the existence, uniqueness and  regularity of solutions of the problem written formally as
\begin{equation}\label{eq:stochnetnonlin}
\left\{\begin{aligned}
\dot{u}_{\me}(t,x) & =  (c_{\me} u_{\me}')'(t,x)-p_{\me}(x) u_{\me}(t,x)+F_{\me}(u_{\me}(t,x)), &x\in (0,\ell_{\me}),\; t\in(0,T],\; \me\in\mE,\;\; & (a)\\
0 & = I_{\mv}U(t,\mv),\;   &t\in(0,T],\; \mv\in\mV,\;\;& (b)\\
\dot{\beta}_{\mv}(t) & = C(\mv)^{\top}U'(t,\mv), & t\in(0,T],\; \mv\in\mV,\;\;& (c)\\
 u_{\me}(0,x) & =  u_{0,\me}(x), &x\in [0,\ell_{\me}],\; \me\in\mE,\;\;& (d)
\end{aligned}
\right.
\end{equation}
where $F_{\me}\colon\real\to\real$ is (globally) Lipschitz and $T>0$ is arbitrary but fixed. However, at the end of the paper we comment on the non-Lipschitz case as well.

Here $\dot{u}_{\me}$ and $u_{\me}'$ denote the time and space derivative, respectively, of $u_{\me}$. For each $\mv\in\mV$, $U(t,\mv)$ and $U'(t,\mv)$ denote the vector of the function values in $\mv$ introduced in \eqref{eq:Fv} and \eqref{eq:Fvv}, respectively, for the function on $G$ defined as
\[u(t,\cdot)=\left(u_{\me}(t,\cdot)\right)_{\me\in\mE}.\]
Analogously to \eqref{eq:Fv}, for each $\mv\in\mV$ and $\mE_{\mv}=\{\me_1,\dots ,\me_{d_{\mv}}\}$ the 1-row matrix in $(\ref{eq:stochnet}c)$ is defined by
\[C(\mv)^{\top}=\left(c_{\me_1}(\mv), \dots ,c_{\me_{d_{\mv}}}(\mv)\right)\in\real^{1\times d_{\mv}}.\]
Equations $(\ref{eq:stochnet}b)$ assume continuity of the function $u(t,\cdot)$ on the metric graph $G$, cf.~\eqref{eq:contv}. 

It is clear by definition, that $(\ref{eq:stochnet}b)$ consists of 
\begin{equation}\label{eq:conteqnumber}
\sum_{\mv\in\mV}(d_{\mv}-1)=2m-n
\end{equation} 
equations. At the same time, $(\ref{eq:stochnet}c)$ consists of $n$ equations. Hence, we have altogether $2m$ (boundary or vertex) conditions in the vertices. 

Let $(\Omega,\mathscr{F},\mathbb{P})$ is a complete probability space endowed with a right-continuous filtration $\mathbb{F}=(\mathscr{F}_t)_{t\in [0,T]}$.
Let the process 
\[(\beta(t))_{t\in [0,T]}=\left(\left(\beta_{\mv}(t)\right)_{t\in [0,T]}\right)_{\mv\in\mV},\] 
be an $\real^{n}$-valued Brownian motion (Wiener process) with covariance matrix 
\[\widetilde{Q}\in \real^{n\times n}\]
with respect to the filtration $\mathbb{F}$; that is,
  $(\beta(t))_{t\in [0,T]}$ is $(\mathscr{F}_t)_{t\in [0,T]}$-adapted and for all $t>s$, $\beta(t)-\beta(s)$ is independent of $\mathscr{F}_s$.

The functions $c_{\me}$ are (variable) diffusion coefficients or conductances, and we assume that
\[c_{\me}\in C[0,\ell_{\me}],\quad 0< c_0\leq c_{\me}(x),\quad\text{for almost all }x\in (0,\ell_{\me}),\text{ for all }  \me\in\mE.\]

The functions $p_{\me}$ are nonnegative, bounded functions, hence

\begin{equation}%\label{eq:}
0\leq p_{\me}\in L^{\infty}(0,\ell_{\me}),\quad \me\in\mE.
\end{equation}
In equation~$(\ref{eq:stochnet}d)$ we pose the initial conditions on the edges.

In Section \ref{sec:stochKircont} we omit the nonlinear drift term in $(\ref{eq:stochnetnonlin}a)$ and perturbe $(\ref{eq:stochnetnonlin}b)$ using a white-noise-term $\dot{\Phi}_{\mv}$ for all $\mv\in\mV$ as, in this case, it turns out that the solution process even to the linear equation is too irregular in space to include a Nemytskij-type nonlinearity. Hence, we investigate the problem
\begin{equation}\label{eq:stochnet}
\left\{\begin{aligned}
\dot{u}_{\me}(t,x) & =  (c_{\me} u_{\me}')'(t,x)-p_{\me}(x) u_{\me}(t,x), &x\in (0,\ell_{\me}),\; t\in(0,T],\; \me\in\mE,\;\; & (a)\\
\dot{\Phi}_{\mv}(t)& = I_{\mv}U(t,\mv),\;   &t\in(0,T],\; \mv\in\mV,\;\;& (b)\\
\dot{\beta}_{\mv}(t) & = C(\mv)^{\top}U'(t,\mv), & t\in(0,T],\; \mv\in\mV,\;\;& (c)\\
 u_{\me}(0,x) & =  u_{0,\me}(x), &x\in [0,\ell_{\me}],\; \me\in\mE.\;\;& (d)
\end{aligned}
\right.
\end{equation}
Here the process 
\[\left(\Theta(t)\right)_{t\in [0,T]}\coloneqq\begin{pmatrix}
	\left(\left(\Phi_{\mv}\right)_{t\in [0,T]}\right)_{\mv\in\mV}\\
	\left(\left(\beta_{\mv}\right)\right)_{t\in [0,T]})_{\mv\in\mV}
\end{pmatrix}\] 
is an $\real^{2m}$-valued Brownian motion (Wiener process) with covariance matrix 
\begin{equation}\label{eq:kovmtx}
Q\in \real^{2m\times 2m}
\end{equation} 
with respect to the filtration $\mathbb{F}$.
The main motivation to consider \eqref{eq:stochnetnonlin} and \eqref{eq:stochnet} is to generalize the classical concept of boundary noise, see, for example, \cite{BZ06, DPZ93} for Gaussian noise and \cite[Sec.~15.1]{PeZa} for L\'evy noise, to quantum graphs, see also \cite{BM10}.

The paper is organized as follows. In Section \ref{sec:determnetwork} we first investigate the linear deterministic version of \eqref{eq:stochnetnonlin} (that is, of \eqref{eq:stochnet}). We rewrite the system in the form of an abstract Cauchy problem governed by the operator $A$ and prove well-posedness by showing that $A$ generates a strongly continuous, analytic contraction semigroup on the Hilbert space $\mcH$ of $L^2$-functions on the edges, see Proposition \ref{prop:eaASttul}. The statement of Proposition \ref{prop:Diropletezik} is the existence and boundedness of the so-called Dirichlet-operator. Using this and results from \cite{DPZ93}, we can treat two versions of the stochastic problem. \\
In Section \ref{sec:stochKir} we consider the problem \eqref{eq:stochnetnonlin}. In Theorem \ref{theo:boundpertKir} we prove the existence and uniqueness of the mild solution for this problem with continuous paths in the Hilbert space $\mcH$. We also verify that this solution is Markov and Feller. Assuming that the vertex values of the eigenfunctions of $A$ are uniformly bounded, we can show in Theorem \ref{theo:eigenfctest} that the mild solution have continuous paths in the fractional domain space of order $\alpha<\frac{1}{4}$. This is the case, for example, if all diffusion coefficients and all edge lengths are constant and equal $1$, see Example \ref{ex:Laplace}. In Remark \ref{rem:polynnonlin} we treat the case of odd-degree polynomial type nonlinearities.\\
In Section \ref{sec:stochKircont} we briefly investigate the problem \eqref{eq:stochnet} where all the boundary conditions are perturbed by some noise, but due to the low regularity shown, only in the linear case. In Theorem \ref{theo:boundpertKircont} we prove that the stochastic convolution process has values in the fractional domain space of $A$ of order $\alpha< -\frac14$ only. Therefore, Nemytskij type nonlinearities, as in the previous section, cannot be considered as pont evaluation is not well-defined anymore. \\
The regularity results obtained in this paper are the analogues of the ones obtained by Da Prato and Zabczyk in \cite{DPZ93} in case of a single interval and classical boundary Dirichlet or Neumann noise.

%%%%%%%%%%%%%%%%%%%%%%%%%%%%%%%%%%%%%%%%%%%%%%%%%%%%%%%%%%%%%%%%%%%%%%%%%%%%%%%%%%%%%%%%%%%%%%%%%%%%%%%%%%%%%%%%%%%%%%%%%%%%%%%%%%%%%%%%%
\section{Heat equation on a network}\label{sec:determnetwork}
%%%%%%%%%%%%%%%%%%%%%%%%%%%%%%%%%%%%%%%%%%%%%%%%%%%%%%%%%%%%%%%%%%%%%%%%%%%%%%%%%%%%%%%%%%%%%%%%%%%%%%%%%%%%%%%%%%%%%%%%%%%%%%%%%%%%%%%%%%%

%%%%%%%%%%%%%%%%%%%%%%%%%%%%%%%%%%%%%%%%%%%%%%%%%%%%%%%%%
\subsection{The abstract Cauchy problem}\label{subsec:systeq}
%%%%%%%%%%%%%%%%%%%%%%%%%%%%%%%%%%%%%%%%%%%%%%%%%%%%%%%%%

We start with the deterministic problem
\begin{equation}\label{eq:netcp}
\left\{\begin{aligned}
\dot{u}_{\me}(t,x) & =  (c_{\me} u_{\me}')'(t,x)-p_{\me}(x) u_{\me}(t,x), &x\in (0,\ell_{\me}),\;t>0,\; \me\in\mE,\;\; & (a)\\
0 & = I_{\mv}U(t,\mv),\;   &t>0,\; \mv\in\mV,\;\;& (b)\\
0 & =  C(\mv)^{\top}U'(t,\mv), & t>0,\; \mv\in\mV,\;\;& (c)\\
 u_{\me}(0,x) & =  u_{0,\me}(x), &x\in [0,\ell_{\me}],\;\; \me\in\mE.\;\;& (d)
\end{aligned}
\right. ,
\end{equation}
where $0$ denotes the constant $0$ vector of dimension $d_{\mv}-1$ on the left-hand-side of $(\ref{eq:netcp}b)$.

We would like to rewrite our system in the form of an abstract Cauchy problem. First we consider the Hilbert space
\begin{equation}\label{eq:E2}
\mcH\coloneqq \prod_{\me\in\mE}L^2\left(0,\ell_{\me}\right)
\end{equation}
as the \emph{state space} of the edges, endowed with the natural inner product
\[\langle f,g\rangle_{\mcH}\coloneqq \sum_{\me\in\mE} \int_0^{\ell_{\me}} f_{\me}(x)g_{\me}(x) dx,\qquad
f=\left(f_{\me}\right)_{\me\in\mE},\;g=\left(g_{\me}\right)_{\me\in\mE}\in \mcH.\]
On $\mcH$ we define the operator
\begin{equation}\label{eq:opAmax}
A_{\max}\coloneqq \diag\left(\frac{d}{dx}\left(c_{\me} \frac{d}{dx}\right)-p_{\me} \right)_{\me\in\mE}
\end{equation}
with maximal domain
\begin{equation}\label{eq:domAmax}
D(A_{\max}) =\mcH^2\coloneqq \prod_{\me\in\mE} H^2(0,\ell_{\me}).
\end{equation}
We also introduce the \emph{boundary space} 
\begin{equation}\label{eq:Y}
\mcY\coloneqq \ell^2(\real^{2m})\cong \real^{2m}.
\end{equation}

Notice that for fixed $u\in D(A_{\max})$, the boundary (or vertex) conditions can be written as
\begin{equation}\label{eq:contKir}
0_{\real^{d_{\mv}-1}}=I_{\mv}U(\mv),\; 0=C(\mv)^{\top}U'(\mv),\quad \mv\in\mV,
\end{equation}
cf.~$(\ref{eq:netcp}b)$ and $(\ref{eq:netcp}c)$.

\begin{remark}
Define $A_{\mv}$ as the square matrix that arises from $I_{\mv}$ by inserting an additional row containing only $0$'s; that is,
\begin{equation}\label{eq:Av}
A_{\mv}=\begin{pmatrix}
	1 & -1 & & \\
	& \ddots & \ddots & \\
	& & 1 & -1 \\
	0 & \hdots & 0 & 0
	\end{pmatrix}\in\real^{d_{\mv}\times d_{\mv}}.
\end{equation}
Furthermore, let $B_{\mv}$ be the square matrix defined by
\begin{equation}\label{eq:Bv}
B_{\mv}=\begin{pmatrix}
	0 & \hdots &  0\\
	\vdots   &  & \vdots\\
	0 & \hdots &  0\\
	c_{\me_1}(\mv) & \hdots &  c_{\me_{d_{\mv}}}(\mv),
	\end{pmatrix}\in\real^{d_{\mv}\times d_{\mv}}.
\end{equation} 
It is straghtforward that for a fixed $\mv\in\mV$, equations \eqref{eq:contKir} have the form
\begin{equation}\label{eq:AvBv}
A_{\mv}U(\mv)+B_{\mv}U'(\mv)=0_{\real^{d_{\mv}}},
\end{equation}
where the $d_{\mv}\times (2d_{\mv})$ matrix $\left(A_{\mv},B_{\mv}\right)$ has maximal rank. Thus, our vertex conditions have the form as in \cite[Sec.~1.4.1]{BeKu}.
\end{remark}

We now define the feedback operator $B\colon D(A_{\max})\to \mcY$ by
\begin{equation}\label{eq:opB}
\begin{split}
D(B) & = D(A_{\max});\\
B u & =\begin{pmatrix}
	\left(I_{\mv}U(\mv)\right)_{\mv\in\mV}\\
	\left(C(\mv)^{\top}U'(\mv)\right)_{\mv\in\mV}
\end{pmatrix},
\end{split}
\end{equation}
where the first ''block'' of $Bu$ has $\sum_{\mv\in\mV}(d_{\mv}-1)=2m-n$ coordinates (see also \eqref{eq:conteqnumber}), while the second ''block'' has $n$ coordinates. Hence, $B$ maps indeed into $\real^{2m}$.

With these notations, we can finally rewrite \eqref{eq:netcp} in form of an abstract Cauchy problem. Define
\begin{align}\label{eq:amain}
A &\coloneqq A_{\max}\\
D(A)&\coloneqq\{u\in D(A_{\max})\colon Bu=0_\mcY\}.
\end{align}
Using this, \eqref{eq:netcp} becomes
\begin{equation}\label{eq:acp}
\left\{\begin{array}{rcll}
\dot{u}(t)&=& A u(t), &t> 0,\\
u(0)&=& u_0,
\end{array}
\right.
\end{equation}
with $u_0=(u_{0,1},\dots ,u_{0,m})^{\top}$.

%%%%%%%%%%%%%%%%%%%%%%%%%%%%%%%%%%%%%%%%%%%%%%%%%%%%%%%%%%%
\subsection{Well-posedness of the abstract Cauchy problem}
%%%%%%%%%%%%%%%%%%%%%%%%%%%%%%%%%%%%%%%%%%%%%%%%%%%%%%%%%%%

\begin{proposition}\label{prop:Diropletezik}
The operators $A_{\max}$ and $B$ satisfy the conditions in \cite[Section 15.1]{PeZa} and \cite[Section 1]{DPZ93}, that is
\begin{enumerate}
	\item the operator $A$ defined in \eqref{eq:amain} generates a $C_0$-semigroup on $\mcH$;
	\item for $\lambda\in\rho(A)$ there exists a bounded operator $D_{B,\lambda}\colon \mcY\to \mcH$ such that
	\begin{equation}\label{eq:Dirop}
	D_{B,\lambda}=\left(B\mid_{\Ker(\lambda-A_{\max})}\right)^{-1}.
	\end{equation}
\end{enumerate}
\end{proposition}
\begin{proof}$ $

1. Follows directly by \cite[Sec.~3.2 and Rem.~2.9]{EK19} or \cite[Rem.~3.6]{MR07}.

2. We will show that the assumptions (a)-(d) of \cite[(1.13)]{Gr87} are satisfied for $A_{\max}$ and $B$. Then, by \cite[Lemma 1.2]{Gr87} the existence of the bounded operator of $D_{B,\lambda}$ in \eqref{eq:Dirop} follows for every $\lambda\in\rho(A).$

The operator $\left(A_{\max},D(A_{\max})\right)$ is densely defined and closed on $\mcH$, hence assumption (a) is satisfied.
The boundary operator as a mapping
\[B\colon \left(D(A_{\max}),\|\cdot \|_{\mcH^{2}}\right)\to \mcY\]
is bounded by Sobolev embedding, which is assumption (b). The statement of (d) is exactly 1.~above. 

It remains only to prove (c); that is, $\mathrm{Im}\,B=\real^{2m}$, which is the assertion of Proposition \ref{prop:Bsurj}. 
\end{proof}

\begin{proposition}\label{prop:eaASttul}
The operator $(-A,D(A))$, where $A$ is defined in \eqref{eq:amain}, is the operator associated with the form
\begin{equation}\label{eq:ea}
\begin{split}
\ea(u,v)&=\sum_{\me\in\mE}\int_0^{\ell_{\me}}c_{\me}(x)u_{\me}'(x)v_{\me}'(x)\dx+\sum_{\me\in\mE}\int_0^{\ell_{\me}} p_{\me}(x)u_{\me}(x)v_{\me}(x)\dx,\\
D(\ea)&=\left\{u\in \mcH^1\colon I_{\mv}U(\mv)=0,\; \mv\in\mV\right\},
\end{split}
\end{equation}
where
\begin{equation}%\label{eq:}
\mcH^1\coloneqq \prod_{\me\in\mE}H^1(0,\ell_{\me}),
\end{equation}
in the following sense: 
\begin{equation}%\label{eq:ea}
\begin{split}
D(A)&=\left\{u\in D(\ea) \colon \exists\, h\in \mcH\text{ s.t. }\ea(u,v)=\langle h,v\rangle_{\mcH}\text{ for all }v\in D(\ea)\right\},\\
-A u&=h.
\end{split}
\end{equation}
The form $(\ea,D(\ea))$ is symmetric, densely defined, continuous, closed and accretive. The operator $(A,D(A))$ is densely defined, dissipative, sectorial and self-adjoint with $(0,+\infty)\subset\rho(A)$.
The strongly continuous semigroup $(S(t))_{t\geq 0}$ generated by $(A,D(A))$ is analytic, positive and contractive. 
\end{proposition}
\begin{proof}
All the statements follow by \cite[Sec.~3]{MR07} (see also \cite[Prop.~2.3--2.6]{KS21}).
\end{proof}

%%%%%%%%%%%%%%%%%%%%%%%%%%%%%%%%%%%%%%%%%%%%%%%%%%%%%%%%%%%%%%%%%%%%%%%%%%%%%%%%%%%%%%%%%%
\section{Stochastic perturbation of the Kirchhoff--Neumann vertex conditions}\label{sec:stochKir}
%%%%%%%%%%%%%%%%%%%%%%%%%%%%%%%%%%%%%%%%%%%%%%%%%%%%%%%%%%%%%%%%%%%%%%%%%%%%%%%%%%%%%%%%%%%%

In this section we consider the problem
\begin{equation}\label{eq:stochnetK}
\left\{\begin{aligned}
\dot{u}_{\me}(t,x) & =  (c_{\me} u_{\me}')'(t,x)-p_{\me}(x) u_{\me}(t,x)+F_{\me}(u_{\me}(t,x)), &x\in (0,\ell_{\me}),\; t\in(0,T],\; \me\in\mE,\;\; & (a)\\
0 & = I_{\mv}U(t,\mv),\;   &t\in(0,T],\; \mv\in\mV,\;\;& (b)\\
\dot{\beta}_{\mv}(t) & = C(\mv)^{\top}U'(t,\mv), & t\in(0,T],\; \mv\in\mV,\;\;& (c)\\
 u_{\me}(0,x) & =  u_{0,\me}(x), &x\in [0,\ell_{\me}],\; \me\in\mE,\;\;& (d)
\end{aligned}
\right.
\end{equation}
where $F_{\me}\colon\real\to\real$ is (globally) Lipschitz, see \eqref{eq:stochnetnonlin}.

To treat this problem, on $\mcH$ we define the maximal operator $(A_{K,\max},D(A_{K,\max}))$ in a slightly different way than it has been done in \eqref{eq:opAmax} and \eqref{eq:domAmax}. Namely, we put the continuity conditions in its domain (see also \cite{KS21}), that is, we set
\begin{align}\label{eq:opAKmax}
A_{K,\max}&\coloneqq \diag\left(\frac{d}{dx}\left(c_{\me} \frac{d}{dx}\right)-p_{\me} \right)_{\me\in\mE},\\
D(A_{K,\max})&\coloneqq \left\{u\in\mcH^2\colon I_{\mv}U(\mv)=0,\; \mv\in\mV \right\}.
\end{align}
Accordingly, we have to modify the boundary space as
\begin{equation}\label{eq:YK}
\mcY_{K}\coloneqq\ell^2(\real^n)\cong \real^n,
\end{equation}
where $n=|\mV|$ the number of vertices in $G$, cf.~\eqref{eq:Y}. The feedback operator becomes $B_K\colon D(A_{K,\max})\to \mcY_K$,

\begin{align}\label{eq:opBK}
D(B_K) & = D(A_{K,\max});\\
B_K u & \coloneqq \left(C(\mv)^{\top}U'(\mv)\right)_{\mv\in\mV},
\end{align}
see \eqref{eq:opB}.

Defining
\begin{align}\label{eq:aKmain}
A &\coloneqq A_{K,\max}\\
D(A)&\coloneqq\{u\in D(A_{K,\max})\colon B_Ku=0_{\mcY_K}\}
\end{align}
we obtain the same operator as in \eqref{eq:amain}.

Mimicking the proof of Proposition \ref{prop:Diropletezik} one obtains the following result.

\begin{proposition}\label{prop:DiropKirletezik}
The operators $A_{K,\max}$ and $B_K$ satisfy the conditions in \cite[Section 15.1]{PeZa} and \cite[Section 1]{DPZ93}, that is
\begin{enumerate}
	\item the operator $A$ defined in \eqref{eq:aKmain} -- that is, in \eqref{eq:amain} --	generates a $C_0$-semigroup on $\mcH$;
	\item for $\lambda>0$ there exists a bounded operator $D_{B_K,\lambda}\colon \mcY_K\to \mcH$ such that
	\begin{equation}\label{eq:DiropK}
	D_{B_K,\lambda}=\left(B_K\mid_{\Ker(\lambda-A_{K,\max})}\right)^{-1}.
	\end{equation}
\end{enumerate}
\end{proposition}

In order to define the so-called mild solution to  \eqref{eq:stochnetK} we first study the stochastic convolution process defined by
\begin{equation}\label{eq:ZKt}
Z_K(t)\coloneqq \int_0^t (\lambda-A) S(t-s)D_K\, d\beta (s), \quad t\in [0,T],
\end{equation}
where $\lambda>0$ is fixed, $D_K\coloneqq D_{B_K,\lambda}$ and \[\beta=\left(\beta_{\mv}\right)_{\mv\in\mV},\] 
is a $\real^{n}$-valued Brownian motion (Wiener process) with covariance matrix 
\[\widetilde{Q}\in \real^{n\times n}.\]
To this aim we first introduce the fractional domain spaces of the generator.
Since $A$ generates a contractive analytic semigroup, we can define its fractional powers for $\la>0$ and $\alpha\in (0,1)$. In particular, the fractional domain spaces
\begin{equation}\label{eq:Halphapos}
\mcH_{\alpha}\coloneqq D((\la-A)^{\alpha}),\quad \|u\|_{\alpha}\coloneqq \|(\la-A)^{\alpha}u\|,\quad u\in D((\la-A)^{\alpha})
\end{equation}
are Banach spaces. We fix $\mcH_0\coloneqq \mcH$.

For $\alpha\in(-1,0)$ we define the extrapolation spaces $\mcH_{\alpha}$ as the completion of $\mcH$ under the norms $\|u\|_{\alpha}\coloneqq \|(\la-A)^{\alpha}u\|$, $u\in \mcH$.

It is well-known (see e.g.~\cite[\textsection II.4--5.]{EN00}) that up to equivalent norms, these spaces are independent of the choice of $\la>0$.

\begin{remark}\label{rem:omega0}
Since the semigroup $(S(t))_{t\geq 0}$ is contractive, hence bounded, then by \cite[Prop.~3.1.7]{Haase06}
we can choose $\la=0$ in \eqref{eq:Halphapos}. That is,
\[\mcH_{\alpha}\cong D((-A)^{\alpha}),\quad \alpha\in [0,1),\]
when $D((-A)^\alpha)$ is equipped with the graph norm.
\end{remark}

\begin{theorem}\label{theo:ZKcont}
The stochastic convolution $Z_K(\cdot)$ given by \eqref{eq:ZKt} is well-defined, 
\begin{equation}%\label{eq:}
Z_K\in C\left([0,T],L^2(\Omega,\mcH)\right)
\end{equation}
and has an $\mcH$-valued continuous version.
\end{theorem}
\begin{proof}
First we prove that the stochastic process $Z_K(\cdot)$ in \eqref{eq:ZKt} is well-defined in $\mcH$, that is,
\begin{equation}\label{eq:ZKtcontpf00}
\int_0^T\left\|(\lambda-A) S(t)D_K\widetilde{Q}^\frac12\right\|^2_{\HS(\mcY_K,\mcH)}\dt<+\infty,
\end{equation}
where $\HS$ denotes the Hilbert--Schmidt-norm between the appropriate spaces. As we have
\[
\left\|(\lambda-A) S(t)D_K\widetilde{Q}^\frac12\right\|^2_{\HS(\mcY_K,\mcH)}\leq \left\|(\lambda-A) S(t)D_K\right\|^2_{\HS(\mcY_K,\mcH)} \cdot\Tr(\widetilde{Q}),
\]
it is enough to verify that
\begin{equation}\label{eq:ZKtcontpf}
\int_0^T\left\|(\lambda-A) S(t)D_K\right\|^2_{\HS(\mcY_K,\mcH)}\dt<+\infty.
\end{equation}
Using a standard energy argument (or, alternatively, Parseval's formula), it follows that
\begin{equation}\label{eq:DKbecsles}
\int_0^T\left\|(\lambda-A) S(t)D_K\right\|^2_{\HS(\mcY_K,\mcH)}\dt\leq c_T\cdot \left\|(\lambda-A)^{\frac{1}{2}}D_K\right\|^2_{\HS(\mcY_K,\mcH)}.
\end{equation}
Since
\[D_K\colon \mcY_K\to D(A_{K,\max})\]
and
\[D(A_{K,\max})=\left\{u\in\mcH^2\colon I_{\mv}U(\mv)=0,\; \mv\in\mV \right\}\hookrightarrow D(\ea),\]
see \eqref{eq:ea}, $D_K$ can be regarded as a boundedn operator into $D(\ea)$.

On the other hand, we obtain that the form
\begin{equation}\label{eq:eaomega}
\ea_{\la}(u,v)\coloneqq \ea(u,v)+\la\cdot \langle u,v\rangle_E,\quad u,v\in D(\ea)
\end{equation}
is coercive, symmetric and continuous, see Proposition \ref{prop:eaASttul} and \cite[Rem.~7.3.3]{Haase06}. It is straightforward that the operator associated with $\ea_{\la}$ is $\la-A$. For the form-domain $D(\ea_{\la})=D(\ea)$, see \eqref{eq:ea}, equipped with the usual $\mcH^1$-norm, we have that
\begin{equation}\label{eq:DA12Dea}
D((\la-A)^{\frac{1}{2}})\cong D(\ea)
\end{equation}
holds with equivalence of norms, see e.g.~\cite[Prop.~5.5.1]{Ar04}. 

Hence $D_K$ is a bounded linear operator into $D((\lambda-A)^{\frac{1}{2}})$. Since the range of $D_K$ is finite dimensional, 
\[\left\|(\lambda-A)^{\frac{1}{2}}D_K\right\|_{\HS(\mcY_K,\mcH)}<+\infty\]
holds. Thus, using \eqref{eq:DKbecsles}, the assertion in \eqref{eq:ZKtcontpf}, hence \eqref{eq:ZKtcontpf00} holds, the stochastic convolution  $Z_K(\cdot)$ is well-defined in $\mcH$. 

Finally, since the semigroup $(S(t))_{t\geq 0}$ is contractive, \cite[Rem.~1]{HauSei01} implies the continuity of the trajectories.
\end{proof}

Next, we define the mild solution of \eqref{eq:stochnetK} as in \cite[Sec.~15.1]{PeZa} (and also \cite[Sec.~1]{DPZ93}). 
Let
\[u(t)(\cdot)=u(t,\cdot)=\left(u_{\me}(t,\cdot)\right)_{\me\in\mE},\quad F=\left(F_{\me}\right)_{\me\in\mE},\quad F(u(s))=\left(F_{\me}(u_{\me}(s))\right)_{\me\in\mE}\].
\begin{definition}
If $X_K(\cdot)\in C\left([0,T],L^2(\Omega,\mcH)\right)$, $\{X_K(t)\}_{t\in [0,T]}$ is $(\mathscr{F}_t)_{t\in [0,T]}$-adapted and  for all $t\in [0,T]$,
\begin{equation}\label{eq:stochDK}
X_K(t)=X_K(t,u_0)=S(t)u_0+\int_0^t S(t-s)F(X_k(s))\ds +Z_k(t),
\end{equation}
$\mathbb{P}$-almost surely, where $Z_K$ is defined by \eqref{eq:ZKt}, then $X_K$ is called a \emph{mild solution} of the problem \eqref{eq:stochnetK}. 
\end{definition}

\begin{theorem}\label{theo:boundpertKir}
For all $u_0\in \mcH$, equation \eqref{eq:stochnetK} has a unique mild solution. The mild solution has an $\mcH$-valued continuous version and it is Markov and Feller.
\end{theorem}
\begin{proof}
Existence and uniqueness of a mild solution of  \eqref{eq:stochnetK}  follows by a simple standard fixed point argument in the Banach space $C\left([0,T],L^2(\Omega,\mcH)\right)$ using Theorem \ref{theo:ZKcont} which assersts that  $Z_K\in C\left([0,T],L^2(\Omega,\mcH)\right)$. Since, also by Theorem \ref{theo:ZKcont}, $Z_K$ has an  $\mcH$-valued continuous version, it can be easily seen in view of  by \eqref{eq:stochDK}, that so does $X_K$.\\
To prove that $X_K$ is Markov, we may use the same reasoning as in the proof of \cite[Thm.~9.21]{DPZbook} with $E=\mcH$ noting that $Z_K$ has a continuous version by Theorem \ref{theo:ZKcont}.\\
Finally, let $u_0,v_0\in\mcH$ arbitrary given initial values. Then, by \eqref{eq:stochDK}, the contractivity of $(S(t))_{t\geq 0}$ and the global Lipschitz continuity of $F$, a simple Gronwall argument shows that
\[\sup_{t\in [0,T]}\left\|X_K(t,u_0)-X_K(t,v_0)\right\|_{\mcH}^2\leq C(T)\cdot \|u_0-v_0\|_{\mcH}^2, \text{ almost surely,}\]
and thus $X_K$ is Feller (c.f. \cite[Rem.~9.33]{PeZa}).
\end{proof}

\begin{remark}\label{rem:Aeigen}
%Since for the constant function $\mathbf{1}$, $\mathbf{1}\in D(A)$ and $A\mathbf{1}=0$ holds, we have $0\in\sigma (A)$. Using that the semigroup $S$ is contractive, we obtain for the spectral bound $s(A)=0$. 
Observe that $D(A)\subset \mcH^1$, and by the Rellich-Kondrachov theorem, $\mcH^1\hookrightarrow \mcH$ is a continuous, compact embedding. If $\lambda\in\rho(A)$ is arbitrary, for the resolvent operator $\ran (R(\lambda,A))\subset D(A)$ holds, and we obtain that $R(\lambda,A)\colon \mcH\to \mcH$ is bounded and compact. That is, $A$ has compact resolvent and thus $A$ has only point spectrum. Since $A$ is self-adjoint and dissipative, its eigenvalues $(\lambda_k)_{k\in\nat}$ form a sequence of negative real numbers and 
\begin{equation}\label{eq:lambdak}
\lambda_k\to-\infty,\quad k\to\infty.
\end{equation}
We may then choose a set $(f_k)_{k\in\nat}\subset D(A)$ of eigenfunctions such that
\begin{equation}\label{eq:fk}
Af_k=\lambda_k f_k,\quad k\in\nat,
\end{equation}
and the functions $(f_k)_{k\in\nat}$ form a complete orthonormal system in $\mcH$.
\end{remark}

\begin{remark}
In the view of \cite[(9)--(11)]{DPZ93} we can also consider adding a space-time white noise term to $(\ref{eq:stochnetK}a)$. To obtain the statements of Theorem \ref{theo:boundpertKir} for this new equation, by \cite[(17)]{DPZ93} it is enough to show that
\begin{equation}%\label{eq:}
\int_0^T\left\|S(t)\right\|_{\HS}^2\dt<\infty.
\end{equation}
Let $(f_k)_{k\in\nat}$ be the complete orthonormal system consisting of eigenfunctions of $A$, see Remark \ref{rem:Aeigen}. Then we have
\begin{align*}
\int_0^T\left\|S(t)\right\|_{\HS}^2\dt&=\int_0^T\sum_{k=1}^{\infty}\left\|S(t)f_k\right\|^{2}\dt=\int_0^T\sum_{k=1}^{\infty}\e^{-2\lambda_k t}\dt\\
&=\int_0^T 1\dt+\int_0^T\sum_{k=2}^{\infty}\e^{-2\lambda_k t}\dt=T+\sum_{k=2}^{\infty}\frac{1}{2\lambda_k}\left(1-\e^{-2\lambda_k T}\right)\\
&\leq T+\sum_{k=2}^{\infty}\frac{1}{2\lambda_k}<\infty
\end{align*}
by Proposition \ref{prop:szigmaA}.
\end{remark}

It turns out that if we assume uniform boundedness for the vertex values of the eigenfunctions $(f_k)$, the mild solution has a continuous version in $\mcH_{\alpha}$ for $\alpha<\frac{1}{4}$.

We introduce now the boundary operator $L\colon D(L)\to \mcY_K$ defined by
\begin{equation}\label{eq:L}
\begin{split}
D(L) & = \left\{u\in \prod_{\me\in\mE}C[0,\ell_{\me}]: I_{\mv} U(\mv)=0,\; \mv\in\mV\right\},\\
L u & \coloneqq\left(u(\mv)\right)_{\mv\in\mV}\in \mcY_K,
\end{split}
\end{equation}
where $u(\mv)$ denotes the common vertex value of the function $u\in D(L)$ on the edges incident to $\mv$. That is, $L$ assigns to each function $u$ that is continuous on $G$ the vector of the vertex values of $u$. Observe that $D(A)\subset D(L)$ holds.

\begin{theorem}\label{theo:eigenfctest}
Suppose that there exists $c>0$ such that 
\begin{equation}\label{eq:fkunifbound}
\|Lf_k\|^2_{\mcY_K} \leq c,\quad k\in\nat,
\end{equation}
where $(f_k)_{k\in\nat}$ is the complete orthonormal system consisting of eigenfunctions of $A$ (see Remark \ref{rem:Aeigen}) and $L$ is the operator \eqref{eq:L}.

Then, for $\alpha<\frac{1}{4}$ the stochastic convolution process $Z_K(\cdot)$ defined in \eqref{eq:ZKt} has a continuous version in $\mcH_{\alpha}$.
\end{theorem}
\begin{proof}
By a straightforward modification of \cite[Thm.~2.3]{DPZ93} to include the covariance matrix $\widetilde{Q}$ (see also,  \cite[Thm.~5.9]{DPZ93}), we have to show that if $\alpha<\frac{1}{4}$, then for a fixed $\lambda>0$ and $D_K= D_{B_K,\lambda}$, there exists $\gamma>0$ such that
\begin{equation}\label{eq:Kproof0}
\int_0^{T}t^{-\gamma}\left\|(\lambda-A)S(t)D_K\widetilde{Q}^\frac12\right\|^2_{\HS(\mcY_K,\mcH_{\alpha})}\dt <+\infty,
\end{equation}
As noted before, 
\[
\left\|(\lambda-A)S(t)D_K\widetilde{Q}^\frac12\right\|^2_{\HS(\mcY_K,\mcH_{\alpha})}\leq \left\|(\lambda-A)S(t)D_K\right\|^2_{\HS(\mcY_K,\mcH_{\alpha})}\cdot\Tr(\widetilde{Q}) 
\]
and thus it is enough to prove that
\begin{equation}\label{eq:Kproof1}
\int_0^{T}t^{-\gamma}\left\|(\lambda-A)S(t)D_K\right\|^2_{\HS(\mcY_K,\mcH_{\alpha})}\dt <+\infty.
\end{equation}
Using Remark \ref{rem:omega0} and proceeding as in \cite[Sec.~3]{DPZ93}, we obtain that \eqref{eq:Kproof1} holds if
\begin{equation}\label{eq:sorosszeg}
\sum_{k=1}^{\infty}(\lambda-\lambda_k)^{2\alpha+\gamma+1}\left\|D_K^*f_k\right\|^2_{\mcY_K}<\infty.
\end{equation}
We will show that under the assumptions,
\begin{equation}\label{eq:DKcsillagbecsles}
\left\|D_K^*f_k\right\|_{\mcY_K}^2 \leq \frac{c}{(\lambda-\lambda_k)^2},\quad k\in\nat,
\end{equation}
where $c\in\real$ is the constant from \eqref{eq:fkunifbound}. By Proposition \ref{prop:szigmaA}, there exist constants $l_1,l_2>0$ such that
\begin{equation}\label{eq:lambdakaszimpt}
l_1\cdot k^2\leq \lambda-\lambda_k\leq l_2\cdot k^2,\quad k\in\nat. 
\end{equation}
Hence, if \eqref{eq:DKcsillagbecsles} holds, by \eqref{eq:lambdakaszimpt} the terms of the series \eqref{eq:sorosszeg} can be estimated as
\[(\lambda-\lambda_k)^{2\alpha+\gamma+1}\left\|D_K^*f_k\right\|^2_{\mcY_K}\leq C \cdot \frac{k^{4\alpha+2\gamma+2}}{k^4} \]
for an appropriate constant $C>0.$ Thus, the series \eqref{eq:sorosszeg} converges if and only if
\[4\alpha+2\gamma+2-4<-1 \Longleftrightarrow 4\alpha+2\gamma<1.\]
This means that there exists appropriate $\gamma>0$ if and only if
\[\alpha<\frac{1}{4}.\] 

Now we turn to the proof of \eqref{eq:DKcsillagbecsles}. First notice that
\begin{equation}\label{eq:DKfk}
\left\|D_K^*f_k\right\|^2_{\mcY_K}=\sum_{i=1}^n \langle e_i,D_K^*f_k \rangle_{\mcY_K}^2,
\end{equation}
where
\[e_i=\left(\begin{smallmatrix}
0\\
\vdots\\
1\\
\vdots\\
0
\end{smallmatrix}\right)\leftarrow i^{\rm th}\hbox{ row}\]
the usual $i$th basis vector in $\real^n$ and $\langle \cdot,\cdot\rangle_{\mcY_K}$ denotes the scalar product in $\mcY_K$. 

By definition, for each $i=1,\dots, n$,
\begin{equation}\label{eq:eiDKcsillag}
\langle e_i,D_K^*f_k \rangle_{\mcY_K} =\langle D_K e_i, f_k \rangle_{\mcH}=\langle u^i, f_k \rangle_{\mcH},
\end{equation}
where $u^i \in D(A_{K,\max})$
\begin{equation}\label{eq:ui}
A_{K,\max}u^i=\lambda u^i,\quad B_Ku^i=e_i,
\end{equation}
see \eqref{eq:DiropK}. Integration by parts then yields
\begin{align}
\langle u^i, f_k \rangle_{\mcH}&=\frac{1}{\lambda}\langle A_{K,\max}u^i, f_k \rangle_{\mcH}\\
& =\frac{1}{\lambda}\sum_{\me\in\mE} \int_0^{\ell_{\me}} \left(c_{\me} (u^i_{\me})'\right)'(x)\cdot f_{k,{\me}}(x)\dx-\frac{1}{\lambda}\sum_{\me\in\mE} \int_0^{\ell_{\me}} p_{\me}(x)\cdot u^i_{\me}(x)\cdot f_{k,{\me}}(x)\dx\\
& =\frac{1}{\lambda}\sum_{\me\in\mE}\left[c_{\me}(x)\cdot (u^i_{\me})'(x)\cdot f_{k,{\me}}(x)\right]_0^{\ell_{\me}}\\
&-\frac{1}{\lambda}\sum_{\me\in\mE}\left( \int_0^{\ell_{\me}} c_{\me}(x)\cdot (u^i_{\me})'(x)\cdot f_{k,{\me}}'(x)\dx+\int_0^{\ell_{\me}} p_{\me}(x)\cdot u^i_{\me}(x)\cdot f_{k,{\me}}(x)\dx\right)\\
& =\frac{1}{\lambda}\left( \sum_{\me\in\mE}\left[c_{\me}(x)\cdot (u^i_{\me})'(x)\cdot f_{k,{\me}}(x)\right]_0^{\ell_{\me}}-\ea (u^i, f_{k})\right)\\
&=\frac{1}{\lambda}\left( \sum_{\me\in\mE}\left[c_{\me}(x)\cdot (u^i_{\me})'(x)\cdot f_{k,{\me}}(x)\right]_0^{\ell_{\me}}+ \langle u^i, Af_{k}\rangle_{\mcH} \right)\\
&=\frac{1}{\lambda}\left( \sum_{\me\in\mE}\left[c_{\me}(x)\cdot (u^i_{\me})'(x)\cdot f_{k,{\me}}(x)\right]_0^{\ell_{\me}}+ \lambda_k\cdot \langle u^i, f_{k}\rangle_{\mcH} \right),
\end{align}
where we have used Proposition \ref{prop:eaASttul} and \eqref{eq:fk}. This implies that
\begin{equation}%\label{eq:}
\langle u^i, f_k \rangle_{\mcH}=\frac{1}{\lambda-\lambda_k}\sum_{\me\in\mE} \left[c_{\me}(x)\cdot (u^i_{\me})'(x)\cdot f_{k,{\me}}(x)\right]_0^{\ell_{\me}}.
\end{equation}
Since $f_k \in D(L)$, we have by similar calculations as in the proof of \cite[Prop.~A.1]{KS21JEE}
\begin{equation}\label{eq:uifk}
\langle u^i, f_k \rangle_{\mcH} =\frac{1}{\lambda-\lambda_k}\langle B_K u^i, Lf_k \rangle_{\mcY_K}=\frac{1}{\lambda-\lambda_k}\langle e_i, Lf_k \rangle_{\mcY_K},
\end{equation}
where we used \eqref{eq:ui}.

Hence by \eqref{eq:fkunifbound}, \eqref{eq:DKfk}, \eqref{eq:eiDKcsillag} and \eqref{eq:uifk} we obtain
\begin{equation}%\label{eq:}
\left\|D_K^*f_k\right\|_{\mcY_K}^2=\frac{1}{(\lambda-\lambda_k)^2}\sum_{i=1}^n \langle e_i, Lf_k \rangle^2_{\mcY_K}=\frac{\|Lf_k\|_{\mcY_K}^2}{(\lambda-\lambda_k)^2}\leq \frac{c}{(\lambda-\lambda_k)^2},
\end{equation}
which is exactly \eqref{eq:DKcsillagbecsles} and the proof is complete.
\end{proof}

\begin{corollary}
Under the hypothesis of Theorem \ref{theo:eigenfctest}, if $u_0\in \mcH_{\alpha}$ for some $\alpha<\frac14$, then the unique mild solution of equation \eqref{eq:stochnetK} has a version with continuous paths in $\mcH_{\alpha}$.
\end{corollary}
\begin{proof}
By Theorem \ref{theo:eigenfctest} $Z_K$ has a version with continuous paths in $\mcH_{\alpha}$ for $\alpha<\frac14$, and in particular, this version is $\mcH$-continuous. For all $\omega\in \Omega$ we then obtain a unique solution $\tilde{X}_K(\cdot,\omega)$ of \eqref{eq:stochDK} using a standard fixed point argument in the Banach space $C([0,T],\mcH)$. By uniqueness, $\tilde{X}_K(\cdot,\omega)$ is a version of the mild solution of \eqref{eq:stochnetK}.  On the right-hand side of \eqref{eq:stochDK} the first and the last terms are elements of $C([0,T],\mcH_\alpha)$ for $\alpha$ for all $\omega\in \Omega$. Finally, the convolution term on the right-hand-side of \eqref{eq:stochDK}, with $X_K$ replaced by $\tilde{X}_K$, is an element of the space $C([0,T],\mcH_\eta)$ by \cite[Lem.~3.6]{vNVW08} for all $\eta<1$, as $\tilde{X}_K(\cdot,\omega)\in C([0,T],\mcH)$, $F$ has linear growth and the semigroup $(S(t))_{t\geq 0}$ is analytic. Therefore, $\tilde{X}_K$ has continuous paths in $\mcH_{\alpha}$.
\end{proof}

\begin{example}\label{ex:Laplace}
Let $c_{\me}\equiv 1$ and $\ell_{\me}\equiv 1$, $\me\in\mE$ in \eqref{eq:stochnetK}. Then \eqref{eq:stochnetK} has a mild solution of the form \eqref{eq:stochDK} with continuous paths in $\mcH_{\alpha}$ for $\alpha<\frac14$.
\end{example}
\begin{proof}
By \cite[Lem.~2.1]{CaPi07} and \cite[Sec.~2]{Ki20}, we have that
\[\|Lf_k\|^2_{\mcY_K}\leq c,\quad k\in \nat,\]
hence \eqref{eq:fkunifbound} holds with an appropriate positive $c$. Thus, by Theorem \ref{theo:eigenfctest}, problem \eqref{eq:stochnetK} has a mild solution of the form \eqref{eq:stochDK} with continuous paths in $\mcH_{\alpha}$ for $\alpha<\frac14$.
%By \cite[(13)]{HaMa20}, for the constants
%\begin{equation}%\label{eq:}
%M\coloneqq 2\max_{x\in [\pi,+\infty)}\frac{1}{1-\frac{\sin x}{x}}\approx 2.29456
%\end{equation}
%and 
%\begin{equation}\label{eq:lmin}
%l\coloneqq\min_{\me\in\mE}\ell_{\me}>0
%\end{equation}
%we have
%\begin{equation}%\label{eq:}
%\|Lf_k\|^2_{\mcY_K}\leq \frac{2\cdot n}{l}\cdot M \eqqcolon c.
%\end{equation}
%Hence, if the diffusion coefficients are all equal to $1$ and there is no potential term, \eqref{eq:fkunifbound} holds with an appropriate positive $c$. Thus, by Theorem \ref{theo:eigenfctest}, problem \eqref{eq:stochnetK} has a mild solution of the form \eqref{eq:stochDK} with continuous paths in $\mcH_{\alpha}$ for $\alpha<\frac14$. 
\end{proof}

\begin{remark}\label{rem:polynnonlin}
In $(\ref{eq:stochnetK}a)$ we may also consider odd-degree polynomial nonlinearities of the form
\[F_{\me}(x)=-x^{2k_{\me}+1}+\sum_{j=0}^{2k_{\me}}a_{\me,j}x^j,\quad \me\in\mE\]
as in \cite[Example after Thm.~4.2]{DPZ93} provided that the hypothesis of Theorem \ref{theo:eigenfctest} is satisfied. Then the mapping $F=\left(F_{\me}\right)_{\me\in\mE}$ satisfies the conditions of \cite[Thm.~4.2]{DPZ93} with
\[K\coloneqq \prod_{\me\in\mE}L^{2(2k_{\me}+1)}(0,\ell_{\me}),\quad J\coloneqq \prod_{\me\in\mE}L^{2(2k_{\me}+1)^2}(0,\ell_{\me}).\]
Using Theorem \ref{theo:eigenfctest} and Sobolev's embedding, we obtain that the stochastic convolution process $Z_K(\cdot)$ has a $J$-continuous version. Therefore, by  \cite[Thm.~4.2]{DPZ93}, equation \eqref{eq:stochnetK}, has a $\mcH$-continuous 
mild solution. Uniqueness in this case follows the same way as in the proof of \cite[Thm.~4.8(i)]{DP04} using the one sided Lipschitz property 
\[\langle F(u)-F(v),u-v\rangle_{\mcH} \leq C \|u-v\|_{\mcH}^2,\,u,v\in \mcH, \,C>0,\]
and the dissipativity of $A$.
\end{remark}

%%%%%%%%%%%%%%%%%%%%%%%%%%%%%%%%%%%%%%%%%%%%%%%%%%%%%%%%%%%%%%%%%%%%%%%%%%%%%%%%%%%%%%%%%
\section{Stochastic perturbation of the vertex conditions}\label{sec:stochKircont}
%%%%%%%%%%%%%%%%%%%%%%%%%%%%%%%%%%%%%%%%%%%%%%%%%%%%%%%%%%%%%%%%%%%%%%%%%%%%%%%%%%%%%%%%%

Here we briefly investigate the problem where all the boundary conditions are perturbed by some noise, but due to the low regularity shown below, only in the linear case. That is, we consider the problem
\begin{equation}\label{eq:stochnetKD}
\left\{\begin{aligned}
\dot{u}_{\me}(t,x) & =  (c_{\me} u_{\me}')'(t,x)-p_{\me}(x) u_{\me}(t,x), &x\in (0,\ell_{\me}),\; t\in(0,T],\; \me\in\mE,\;\; & (a)\\
\dot{\Phi}_{\mv}(t)& = I_{\mv}U(t,\mv),\;   &t\in(0,T],\; \mv\in\mV,\;\;& (b)\\
\dot{\beta}_{\mv}(t) & = C(\mv)^{\top}U'(t,\mv), & t\in(0,T],\; \mv\in\mV,\;\;& (c)\\
 u_{\me}(0,x) & =  u_{0,\me}(x), &x\in [0,\ell_{\me}],\; \me\in\mE,\;\;& (d)
\end{aligned}
\right.
\end{equation}
see \eqref{eq:stochnet}. The process 
\[\Theta\coloneqq\begin{pmatrix}
	\left(\Phi_{\mv}\right)_{\mv\in\mV}\\
	\left(\beta_{\mv}\right)_{\mv\in\mV}
\end{pmatrix}.\] 
is  a $\real^{2m}$-valued Brownian motion (Wiener process) with covariance matrix 
\begin{equation}\label{eq:kovmtxKD}
Q\in \real^{2m\times 2m}.
\end{equation} 

%For arbitrary $\lambda>0$, 
%\begin{equation}\label{eq:stochD}
%X(t)=S(t)u_0+\int_0^t(\lambda-A)S(t-s)D_{B,\lambda}\, d\Theta (s)
%\end{equation}
%is the mild solution of \eqref{eq:stochnetKD}. We remark that somewhat informally, $X$ solves the stochastic equation
%\begin{equation}%\label{eq:}
%dX=AX\dt+ (\lambda-A)D_{B,\lambda}\, d\Theta(t).
%\end{equation}
For the sake of simplicity we fix again $\lambda>0$, denote $D\coloneqq D_{B,\lambda}$, and analogously to \eqref{eq:ZKt} we define the stochastic convolution process by
\begin{equation}\label{eq:Zt}
Z(t)\coloneqq \int_0^t (\lambda-A) S(t-s)D\, d\Theta(t) . 
\end{equation}

We next prove that $Z(\cdot)$ has values in the fractional domain space of $A$ of order $\alpha< -\frac14$ only. Therefore, Nemytskij type nonlinearities, as in the previous section, cannot be considered as pont evaluation is not well-defined anymore.

\begin{theorem}\label{theo:boundpertKircont}
The stochastic convolution process $Z(\cdot)$ has continuous paths in $\mcH_{\alpha}$ for $\alpha<-\frac{1}{4}$.
\end{theorem}
\begin{proof}
%It is enough to show that the stochastic process $Z(\cdot)$ defined in \eqref{eq:Zt} has a continuous version in $\mcH_{\alpha}$. 
Let $\delta>\frac{1}{4}$ arbitrary fixed and we will prove the result for $\alpha\coloneqq -\delta$. Without loss of generality we may assume that $\delta$ is close to $\frac{1}{4}$. 

By a straightforward modification of \cite[Thm.~2.3]{DPZ93} to include the covariance matrix $Q$, see \eqref{eq:kovmtxKD}, we have to show that for some $\gamma>0$,
\begin{equation}\label{eq:proof0}
\int_0^{T}t^{-\gamma}\left\|(\lambda-A)S(t)DQ^{\frac12}\right\|^2_{\HS(\mcY,\mcH_{-\delta})}\dt <+\infty,
\end{equation}
where $\HS$ denotes the Hilbert--Schmidt-norm between the appropriate spaces (see also, \cite[Theorem 5.9]{DPZ93}). As 
\[
\left\|(\lambda-A)S(t)DQ^{\frac12}\right\|^2_{\HS(\mcY,\mcH_{-\delta})}\leq \left\|(\lambda-A)S(t)D\right\|^2_{\HS(\mcY,\mcH_{-\delta})} \cdot\Tr(Q),
\]
it suffices to prove that
\begin{equation}\label{eq:proof1}
\int_0^{T}t^{-\gamma}\left\|(\lambda-A)S(t)D\right\|^2_{\HS(\mcY,\mcH_{-\delta})}\dt <+\infty.
\end{equation}
First we will prove that the Dirichlet operator $D$ maps $\mcY$ into $\mcH_{\rho}$ for $0<\rho<\frac{1}{4}$. By \eqref{eq:DA12Dea} and  Remark \ref{rem:omega0}, we have
\begin{equation}\label{eq:DA12Dea2}
D(\ea)\cong D((\la-A)^{\frac{1}{2}})\cong D((-A)^{\frac{1}{2}})
\end{equation}
with equivalent norms. Using this, similarly as in \cite[Lem.~3.6]{KS21} one can show that
\begin{equation}%\label{eq:}
D((-A)^{\frac{1}{2}})\cong \prod_{\me\in\mE}H_0^1(0,\ell_{\me})\times\real^n.
\end{equation}
By \cite[Sec.~4.3.3]{Triebel78} we have that if $\theta<\frac{1}{2}$, for the complex interpolation space
\begin{equation}\label{eq:DA12interpol}
\left[D((-A)^{\frac{1}{2}}),\mcH\right]_{\theta}\cong \prod_{\me\in\mE}H^{\theta}(0,\ell_{\me}) \times\{0_{\real^n}\}
\end{equation}
holds -- see also \cite{BMZ08} --, where we have used $\mcH\cong\mcH\times \{0_{\real^n}\}$ and \cite[Thm.~4.2.2]{BeLo}.
Furthermore, \cite[Thm.~6.6.9]{Haase06} implies that for $0<\theta<1$,
\begin{equation}\label{eq:interpol1}
\left[D((-A)^{\frac{1}{2}}),\mcH\right]_{\theta}\cong \left[D(A),\mcH\right]_{\frac{\theta}{2}}.
\end{equation}
Using \cite[Thm.~in $\mathsection$4.7.3]{Ar04} and \cite[Prop.~in $\mathsection$4.4.10]{Ar04}, we obtain that for any $0<\rho<1$,
\begin{equation}\label{eq:interpol2}
\left[D(A),\mcH\right]_{\rho}\cong  D((-A)^{\rho})=\mcH_{\rho}.
\end{equation}
Combining \eqref{eq:DA12interpol}, \eqref{eq:interpol1} and \eqref{eq:interpol2} yields that for $0<\rho<\frac{1}{4}$
\begin{equation}\label{eq:Hrhoizom}
\prod_{\me\in\mE}H^{2\rho}(0,\ell_{\me}) \times \{0_{\real^n}\}\cong \mcH_{\rho}
\end{equation}
holds. Hence, by \eqref{eq:domAmax} and \eqref{eq:Hrhoizom} we have that for $0<\rho<\frac{1}{4}$
\begin{equation}\label{eq:DmapsDA14}
D\colon \mcY\to D(A_{\max})\hookrightarrow \prod_{\me\in\mE}H^{2\rho}(0,\ell_{\me})\cong\prod_{\me\in\mE}H^{2\rho}(0,\ell_{\me}) \times \{0_{\real^n}\}\cong \mcH_{\rho}.
\end{equation}
Now we are in the position to prove \eqref{eq:proof1}. We take a small $\ve>0$, to be specified later, and estimate the integral in the following way, where we use the analyticity of the semigroup $(S(t))_{t\geq 0}$ from Proposition \ref{prop:eaASttul},
\begin{align}%\label{eq:}
&\int_0^{T}t^{-\gamma}\left\|(\lambda-A)S(t)D\right\|^2_{\HS(\mcY,\mcH_{-\delta}))}\dt \leq c\cdot \int_0^{T}t^{-\gamma}\left\|(\lambda-A)^{1-\delta}S(t)D\right\|^2_{\HS(\mcY,\mcH)}\dt\notag\\
& = c\cdot\int_0^T t^{-\gamma}\left\|(\lambda-A)^{\ve}S(t)(\lambda-A)^{1-\delta-\ve}D\right\|^2_{\HS(\mcY,\mcH)}\dt\notag\\
&\leq C_T\cdot\int_0^T t^{-\gamma-2\ve}\dt\cdot \left\|(\lambda-A)^{1-\delta-\ve}D\right\|^2_{\HS(\mcY,\mcH)}.\label{eq:intbecsles}
\end{align}
In the last expression,
\begin{equation}\label{eq:tkitevo}
\int_0^T t^{-\gamma-2\ve}\dt<\infty \Longleftrightarrow \gamma+2\ve<1 \Longleftrightarrow \ve<\frac{1}{2}-\frac{\gamma}{2}.
\end{equation}
Since the range of $D$ is finite dimensional, if we take a suitable $\ve>0$ such that $(\lambda-A)^{1-\delta-\ve}D$ is a bounded operator, then the Hilbert-Schmidt norm in \eqref{eq:intbecsles} is finite. To satisfy this, by \eqref{eq:DmapsDA14},
\begin{equation}\label{eq:deltaepszilon}
0<1-\delta-\ve<\frac{1}{4}
\end{equation}
should hold. On the other hand, by \eqref{eq:tkitevo},
\[1-\delta-\ve>\frac{1}{2}+\frac{\gamma}{2}-\delta.\]
Hence, we have to find $\gamma>0$ such that
\[\frac{1}{2}+\frac{\gamma}{2}-\delta<\frac{1}{4}\Longleftrightarrow  \frac{1}{4}+\frac{\gamma}{2}<\delta\]
holds true. Since, by assumption, $\delta>\frac{1}{4}$, an appropriate $0<\gamma<1$ can be chosen. Taking any $\ve>0$ satisfying \eqref{eq:tkitevo} and \eqref{eq:deltaepszilon}, that is
\[\frac{3}{4}-\delta<\ve<\frac{1}{2}-\frac{\gamma}{2},\] 
the estimate \eqref{eq:intbecsles} yields a finite bound for the left-hand-side of \eqref{eq:proof1}. This means that for $\gamma$ chosen this way, \eqref{eq:proof1}, hence \eqref{eq:proof0} is satisfied, and the proof is complete.
\end{proof}

%%%%%%%%%%%%%%%%%%%%%%%%%%%%%%%%% APPENDIX %%%%%%%%%%%%%%%%%%%%%%%%%%%%%%%%%%%%%%%

\appendix

\section{Surjectivity of the boundary operator \texorpdfstring{$B$}{B}}\label{sec:appDir}

In this section we complete the proof of Theorem \ref{prop:Diropletezik}.

\begin{proposition}\label{prop:Bsurj}
For the operator $(B,D(B))$ defined in \eqref{eq:opB}, $\mathrm{Im}\,B=\real^{2m}$ holds.
\end{proposition}
\begin{proof}
We show that for $z\in\real^{2m}$ given arbitrarily, there exists $u\in D(A_{\max})=\mcH^2$ (see \eqref{eq:domAmax}) such that 
\begin{equation}\label{eq:Buz}
Bu=z.
\end{equation}
We will seek $u\in\mcH^2$ in the form 
\begin{equation}\label{eq:uinAmax}
u(x)=\begin{pmatrix}
	\alpha_1\e^{-\gamma x}+\beta_1 \e^{-\gamma(\ell_{1}- x)}\\
	\vdots\\
	\alpha_m\e^{-\gamma x}+\beta_m \e^{-\gamma(\ell_{m}- x)}
\end{pmatrix}
\end{equation}
for suitable vectors $\alpha=\left(\alpha_1,\dots ,\alpha_m\right)^{\top}$, $\beta=\left(\beta_1,\dots ,\beta_m\right)^{\top}$ and constant $\gamma>0$, where $\ell_1,\dots ,\ell_m$ denote the edges lengths in the graph $G$. 

We introduce the notation 
\begin{equation}%\label{eq:}
u(\ell)\coloneqq \begin{pmatrix}
	u(\ell_{1})\\
	\vdots\\
	u(\ell_{m})
\end{pmatrix},\quad
u'(\ell)\coloneqq \begin{pmatrix}
	u'(\ell_{1})\\
	\vdots\\
	u'(\ell_{m})
\end{pmatrix},\quad u\in \mcH^2,
\end{equation}
and
\begin{equation}%\label{eq:}
z_C\coloneqq \begin{pmatrix}
	z_1\\
	\vdots \\
	z_{2m-n}
\end{pmatrix},\quad z_K\coloneqq \begin{pmatrix}
	z_{2m-n+1}\\
	\vdots \\
	z_{2m}
\end{pmatrix}
\end{equation}
for the two ,,blocks'' of the vector $z\in\real^{2m}$.
Using the idea of the proof of \cite[Prop.~3.2]{KrPu20}, there exist $(2m-n)\times m$ matrices $V_0$ and $V_1$ and $n\times m$ matrices $W_0$ and $W_1$ such that equation \eqref{eq:Buz} can be rewritten as
\begin{equation}\label{eq:Buzrew}
V_0u(0)+V_1 u(\ell)=z_C,\quad W_0u'(0)-W_1 u'(\ell)=z_K.
\end{equation}
A straightforward computation shows that for the function in \eqref{eq:uinAmax}, equations \eqref{eq:Buzrew} turn into
\begin{equation}\label{eq:Buzrewu}
\begin{split}
\left(V_0+ V_1 E_{\gamma}\right)\cdot \alpha+\left(V_0E_{\gamma}+V_1\right)\cdot \beta &=z_C,\\
\left(-\gamma W_0+\gamma W_1E_{\gamma}\right)\cdot \alpha+\left(\gamma W_0E_{\gamma}-\gamma W_1\right)\cdot \beta &=z_K
\end{split}
\end{equation}
where $E_{\gamma}$ is the diagonal matrix
\begin{equation}%\label{eq:}
E_{\gamma}\coloneqq \begin{pmatrix}
\e^{-\gamma \ell_1} &  & \\
& \ddots & \\
 &  & \e^{-\gamma \ell_m }
\end{pmatrix} 
\end{equation}
We rewrite now \eqref{eq:Buzrewu} with $(2m)\times (2m)$ block-matrices as
\begin{equation}\label{eq:Buzblock}
\begin{pmatrix}
	V_0 & V_1 \\
	-\gamma W_0 & -\gamma W_1
\end{pmatrix}\cdot \begin{pmatrix}
	\alpha \\
	\beta
\end{pmatrix} + \begin{pmatrix}
	V_1 & V_0 \\
	\gamma W_1 & \gamma W_0
\end{pmatrix}\cdot \begin{pmatrix}
	E_{\gamma} & 0 \\
	0 & E_{\gamma}
\end{pmatrix}\cdot \begin{pmatrix}
	\alpha \\
	\beta
\end{pmatrix}=z.
\end{equation}
Denoting by
\begin{equation}%\label{eq:}
N_{\gamma}\coloneqq \begin{pmatrix}
	V_0 & V_1 \\
	-\gamma W_0 & -\gamma W_1
\end{pmatrix},\quad \widetilde{N}_{\gamma}\coloneqq \begin{pmatrix}
	V_1 & V_0 \\
	\gamma W_1 & \gamma W_0
\end{pmatrix},\quad F_{\gamma}\coloneqq \begin{pmatrix}
	E_{\gamma} & 0 \\
	0 & E_{\gamma}
\end{pmatrix},
\end{equation}
we have to show that
\begin{equation}\label{eq:BuIm}
\mathrm{Im}\, \left(N_{\gamma}+\widetilde{N}_{\gamma}\cdot F_{\gamma}\right)=\real^{2m}.
\end{equation}
First we show that $N_{\gamma}$ is invertible. Using again ideas from the proof of \cite[Prop.~3.2]{KrPu20}, we can permute rows and columns of $N_{\gamma}$ such that we obtain a block diagonal matrix $M_{\gamma}\in\real^{(2m)\times (2m)}$, consisting of $n$ blocks of size $d_{\mv}\times d_{\mv}$ for $\mv\in\mV$. Denoting by $M_{\mv}$ the block corresponding to vertex $\mv$ in $M_{\gamma}$ we have that if $d_{\mv}=1$ then $M_{\mv}=-\gamma\cdot c_{\me}(\mv)$ for $\mE_{\mv}=\{\me\}$. Otherwise,
\begin{equation*}
M_{\mv}=\begin{pmatrix}
	1 & -1 & & \\
	& \ddots & \ddots & \\
	& & 1 & -1 \\
	-\gamma\cdot c_{\me_1}(\mv) & \hdots & \hdots & -\gamma\cdot c_{\me_{d_{\mv}}}(\mv)
\end{pmatrix}
\end{equation*}
for $\mE_{\mv}=\{\me_1,\dots ,\me_{d_{\mv}}\}$.
A straightforward computation yields that
\begin{equation}\label{eq:detMv}
\det M_{\mv}=-\gamma\cdot\left(c_{\me_1}(\mv)+\cdots +c_{\me_{d_{\mv}}}(\mv)\right)\neq 0
\end{equation}
because of the assumption on the $c_{\me}$'s. Hence, we obtain that
\[\det M_{\gamma}=\prod_{\mv\in\mV}\det M_{\mv}=(-\gamma)^n\cdot K_c\neq 0\]
with
\[K_c=\prod_{\substack{\mv\in\mV\\ \mE_{\mv}=\{\me_1,\dots ,\me_{d_{\mv}}\}}}\left(c_{\me_1}(\mv)+\cdots +c_{\me_{d_{\mv}}}(\mv)\right).\]
Hence, $M_{\gamma}$ is invertible. Since permutations do not change the determinant of a matrix, we also have that $N_{\gamma}$ is invertible. That is, in \eqref{eq:BuIm} we have
\begin{equation}\label{eq:Bsurjbiz1}
N_{\gamma}+\widetilde{N}_{\gamma}\cdot F_{\gamma} =N_{\gamma}\cdot\left(\mathrm{Id}+N_{\gamma}^{-1}\cdot \widetilde{N}_{\gamma}\cdot F_{\gamma}\right)
\end{equation}
with $\mathrm{Id}=\mathrm{Id}_{(2m)\times (2m)}$. If for an appropriate $\gamma>0$, 
\begin{equation}\label{eq:Bsurjbiz2}
\left\|N_{\gamma}^{-1}\cdot \widetilde{N}_{\gamma}\cdot F_{\gamma}\right\|_{\max}<1
\end{equation} 
is satisfied, then the matrix \eqref{eq:Bsurjbiz1} is invertible, hence \eqref{eq:BuIm} holds.

First we estimate the max-norm of $N_{\gamma}^{-1}$, or, which is the same, the max-norm of $M_{\gamma}^{-1}$. Clearly, $M_{\gamma}^{-1}$ is the block-diagonal matrix of blocks $M_{\mv}^{-1}$. Each cofactor $M_{k,l}$ of $M_{\mv}$ is the determinant of a matrix of the same type as $M_{\mv}$ itself (but having dimension $(d_{\mv}-1)\times (d_{\mv}-1)$), expect those in the last row of $M_{\mv}^{-1}$ which are all equal to $1$. Thus we obtain by \eqref{eq:detMv} that for $\gamma$ big enough,
\[\left\|M_{\mv}^{-1}\right\|_{\max}=\frac{1}{|\det M_{\mv}|}\cdot \max_{k,l=1,\dots ,d_{\mv}}|M_{k,l}|=\frac{\gamma\cdot L_{c,\mv}}{\gamma\cdot\left(c_{\me_1}(\mv)+\cdots c_{\me_{d_{\mv}}}(\mv)\right)}\eqqcolon K_{c,\mv}\]
with constants $L_{c,\mv}$, $K_{c,\mv}>0.$ Hence, for $\gamma$ big enough,
\begin{equation}\label{eq:Mgammainvnorm}
\left\|N_{\gamma}^{-1}\right\|_{\max}=\left\|M_{\gamma}^{-1}\right\|_{\max}=\max_{\mv\in\mV}\left\|M_{\mv}^{-1}\right\|_{\max}=\max_{\mv\in\mV}K_{c,\mv}\eqqcolon K_c.
\end{equation}

Similarly as above, we can permute rows and columns of $\widetilde{N}_{\gamma}$ such that we obtain a block diagonal matrix $\widetilde{M}_{\gamma}$ consisting of $1\times 1$ blocks $\{\gamma\cdot c_{\me}(\mv)\}$ and blocks of size $d_{\mv}\times d_{\mv}$ (for $d_{\mv}>1$)
\begin{equation*}
\widetilde{M}_{\mv}=\begin{pmatrix}
	1 & -1 & & \\
	& \ddots & \ddots & \\
	& & 1 & -1 \\
	\gamma\cdot c_{\me_1}(\mv) & \hdots & \hdots & \gamma\cdot c_{\me_{d_{\mv}}}(\mv)
\end{pmatrix}.
\end{equation*}
Hence, if $\gamma$ is big enough,
\begin{equation}\label{eq:Mgammahullamnorm}
\left\|\widetilde{N}_{\gamma}\right\|_{\max}=\left\|\widetilde{M}_{\gamma}\right\|_{\max}=\gamma\cdot\max_{\substack{\mv\in\mV\\ \me\in\mE_{\mv}}}\{c_{\me}(\mv)\}\eqqcolon \gamma\cdot\widetilde{K}_c.
\end{equation}
Equations \eqref{eq:Mgammainvnorm} and \eqref{eq:Mgammahullamnorm} imply that
we have
\[\left\|N_{\gamma}^{-1}\cdot \widetilde{N}_{\gamma}\cdot F_{\gamma}\right\|_{\max}\leq \e^{-\gamma \cdot l}\cdot \gamma \cdot K_c\cdot \widetilde{K}_c<1,\text{ if }\gamma\text{ is big enough,}\]
where $l$ is dfined as
\begin{equation}%\label{eq:lmin}
l\coloneqq\min_{\me\in\mE}\ell_{\me}>0
\end{equation}
Thus for $\gamma$ big enough \eqref{eq:Bsurjbiz2} holds. This implies that \eqref{eq:BuIm} is true; that is, for arbitrary $z\in\real^{2m}$ there exist $\gamma>0$, $\alpha, \beta\in\real^m$ such that \eqref{eq:Buzblock} is satisfied. Hence, by defining $u$ as in \eqref{eq:uinAmax} with the constants $\gamma$, $\alpha_j$, $\beta_j$, $j=1,\dots ,m$, we obtain
\[u\in D(A_{\max})\text{ and }Bu=z,\]
and the proof is complete.
\end{proof}

\section{Asymptotics of the spectrum of \texorpdfstring{$A$}{A}}

\begin{proposition}\label{prop:szigmaA}
Let $(\lambda_k)_{k\in\nat}$ be the sequence of eigenvalues of the generator $(A,D(A))$ from Remark \ref{rem:Aeigen}. Then for any $\lambda>0$ there exist constants $l_1,l_2>0$ such that
\begin{equation}\label{eq:lambdakaszimptproof}
l_1\cdot k^2\leq \lambda-\lambda_k\leq l_2\cdot k^2,\quad k\in\nat. 
\end{equation}
\end{proposition}
\begin{proof}
Let us fix $\la>0$ and recall that $\lambda-A$ is the operator associated with the coercive, symmetric, continuous form $\ea_{\la}$ defined in \eqref{eq:eaomega} by
\begin{equation}%\label{eq:ea}
\begin{split}
\ea_{\la}(f,g)&=\ea(f,g)+\la\cdot\langle f,g\rangle_{\mcH},\\
D(\ea_{\la})&=\left\{u\in \mcH^1\colon I_{\mv}U(\mv)=0,\; \mv\in\mV \right\}.
\end{split}
\end{equation}
Let $\mv_1\in\mV$ arbitrary, and define the form $\left(\ea^{1},D(\ea^{1})\right)$ acting in the same way as $\ea_{\la}$ but having Dirichlet condition in the vertex $\mv_1$ in its domain. That is,
\begin{equation}
\begin{split}
\ea^{1}(f,g)&=\ea_{\la}(f,g),\\
D(\ea^{1})&=\left\{u\in \mcH^1 \colon I_{\mv}U(\mv)=0,\;\mv\neq \mv_1,\; U(\mv_1)=0_{\real^{d_{\mv_1}}}\; \right\},
\end{split}
\end{equation}
cf.~\eqref{eq:Fv}. Denote by $(A_1,D(A_1))$ the operator associated with $\left(\ea^{1},D(\ea^{1})\right)$. We can now carry out the proof of \cite[Thm.~3.1.8]{BeKu} applied to the forms $\ea_{\la}$ and $\ea^1$. We only have to use the facts that $D(\ea^{1})$ is a subspace of co-dimension $1$ of $D(\ea_{\la})$, $\ea_{\la}$ and $\ea^{1}$ agree on $D(\ea^{1})$, and the ''min-max principle'' for coercive, symmetric forms holds, see e.g.~\cite[Thm.~6.5]{LaTh}. Thus we obtain that denoting by $(\lambda^1_k)_{k\in\nat}$ the eigenvalues of $(A_1,D(A_1))$ labelled in non-increasing order,
\begin{equation}\label{eq:}
\la-\lambda_{k+1}\leq \lambda^1_k\leq \la-\lambda_{k},\quad k\in\nat
\end{equation}
holds. Continuing this process by defining the finite sequence of symmetric and accretive forms $\ea^i$, $i=1,\dots ,n$ for an ordering of the vertices $\mV=\{\mv_1,\dots ,\mv_{n}\}$ such that
\begin{equation}
\begin{split}
\ea^{i+1}(f,g)&=\ea^{i}(f,g),\\
D(\ea^{i+1})&=\left\{u\in D(\ea^i) \colon U(\mv_{i+1})=0\right\},
\end{split}
\end{equation}
we can apply the proof of \cite[Thm.~3.1.8]{BeKu} for each pair of forms $\ea^i$ and $\ea^{i+1}$. Hence, if we denote by $(\lambda^{i}_k)_{k\in\nat}$ the eigenvalues of $(A_{i},D(A_{i}))$, the operator associated to $(\ea^i,D(\ea^i))$, labelled in non-increasing order, we obtain that
\begin{equation}\label{eq:lambdaik}
\lambda^{i}_{k+1}\leq \lambda^{i+1}_k\leq \lambda^{i}_{k},\quad k\in\nat,\quad i=1,\dots n-1.
\end{equation}
Clearly,
\begin{equation}%\label{eq:ea}
\begin{split}
\ea^{n}(f,g)&=\ea_{\la}(f,g),\\
D(\ea^{n})&=\prod_{\me\in\mE}H_0^1(0,\ell_{\me}).
\end{split}
\end{equation}
It is straightforward that the operator $(A_{n},D(A_{n}))$ associated with $\ea^n$ is the operator acting as $\la-A$ with Dirichlet conditions in all vertices. Using that $c_{\me}>0$, $p_{\me}\geq 0$, $\me\in\mE$, and the min-max princible holds, we have that for the set $\left(\lambda^n_k\right)_{k\in\nat}$ of eigenvalues of $A_{n}$ there exist constants $l_1,l_2>0$ such that
\begin{equation}%\label{eq:lambdakaszimpt}
l_1\cdot k^2\leq \lambda^n_k\leq l_2\cdot k^2,\quad k\in\nat,
\end{equation}
see also \cite[Probl.~6.1]{LaTh}. Thus, by \eqref{eq:lambdaik} also \eqref{eq:lambdakaszimptproof} holds which finishes the proof.
\end{proof}
\textbf{Acknowledgements.} M.~Kovács acknowledges the support of the Marsden Fund of the Royal Society of New
Zealand through grant no.~18-UOO-143, the Swedish Research Council (VR) through grant
no.~2017-04274 and the National Research, Development, and Innovation Fund of Hungary under Grant no. TKP2021-NVA-02 and Grant no. K-131545.

E.~Sikolya was supported by the OTKA grant no.~135241.

The authors would like to thank the anonymous referee for the careful reading of the manuscript and for the useful comments that
helped them to improve the presentation and the results of the paper significantly.

\providecommand{\bysame}{\leavevmode\hbox to3em{\hrulefill}\thinspace}
\providecommand{\MR}{\relax\ifhmode\unskip\space\fi MR }
% \MRhref is called by the amsart/book/proc definition of \MR.
\providecommand{\MRhref}[2]{%
  \href{http://www.ams.org/mathscinet-getitem?mr=#1}{#2}
}
\providecommand{\href}[2]{#2}

\end{document}